\documentclass[orivec,letterpaper,envcountsect,11pt]{llncs}

\usepackage{float}
\usepackage{comment}
\usepackage{color}
\usepackage{graphicx}
\usepackage{amssymb, amsmath, amsfonts}
\usepackage{bbm}

\usepackage[noend]{algorithmic}
\usepackage{comment}
\usepackage{xifthen}

\usepackage[singlespacing]{setspace}
\usepackage[papersize={8.5in,11in},left=1in,right=1in,top=1in,bottom=1in]{geometry}
\usepackage{nicefrac}

\usepackage[colorlinks=true,linkcolor=blue]{hyperref}
\usepackage[capitalise,noabbrev]{cleveref}
\usepackage[all]{hypcap}
\usepackage{bm}

\renewenvironment{proof}[1][]
    {\noindent
       \ifx&#1&{\it Proof.}
       \else{\it Proof ({#1}).}
       \fi}{\hfill $\blacksquare$}

\DeclareMathOperator{\E}{E}

\spnewtheorem{fact}{Fact}{\bfseries}{}

\newcommand{\HH}{H}

\newfloat{algorithm}{tb}{log}
\floatname{algorithm}{Algorithm}

\newcommand{\lbinom}[2]{\tbinom{#1}{\downarrow #2}}

\newcommand{\assign}[1]{\hat{#1}}

\newcommand{\com}[1]{\begin{flushright}\small{$\vartriangleright$#1}\end{flushright}}

\newcommand{\fullhline}{~\\\begin{tabular}{p{\textwidth}}\hline \hspace{\textwidth}\end{tabular}\\[-2.1ex]}

\title{Improved Algorithms for Solving Polynomial Systems over GF(2) by Multiple Parity-Counting}

\author{Itai Dinur}

\institute{Department of Computer Science, Ben-Gurion University, Israel}

\begin{document}

\maketitle

\begin{abstract}

We consider the problem of finding a solution to a multivariate polynomial equation system of degree $d$ in $n$ variables over $\mathbb{F}_2$. For $d=2$, the best-known algorithm for the problem is by Bardet et al. [\textit{J. Complexity}, 2013] and was shown to run in time $O(2^{0.792n})$ under assumptions that were experimentally found to hold for random equation systems. The best-known worst-case algorithm for the problem is due to Bj\"{o}rklund et al. [ICALP'19]. It runs in time $O(2^{0.804n})$ for $d = 2$ and $O(2^{(1 - 1/(2.7d))n})$ for~$d > 2$.

In this paper, we devise a worst-case algorithm that improves the one by Bj\"{o}rklund et al. It runs in time $O(2^{0.6943n})$ (or $O(1.6181^n)$) for $d = 2$ and $O(2^{(1 - 1/(2d))n})$ for $d > 2$. Our algorithm thus outperforms all known worst-case algorithms, as well as ones analyzed for random equation systems. We also devise a second algorithm that outputs all solutions to a polynomial system and has similar complexity to the first (provided that the number of solutions is not too large).

A central idea in the work of Bj\"{o}rklund et al. was to reduce the problem of finding a solution to a polynomial system over $\mathbb{F}_2$ to the problem of counting the parity of all solutions. A parity-counting instance was then reduced to many smaller parity-counting instances. Our main observation is that these smaller instances are related and can be solved more efficiently by a new algorithm to a problem which we call \emph{multiple parity-counting}.

\end{abstract}

\thispagestyle{empty}

\newpage
\setcounter{page}{1}

\section{Introduction}

We study the problem of solving a system of multivariate
polynomial equations over the field $\mathbb{F}_2$, which is a fundamental problem in
computer science. The input to this problem consists of a system of $m$ polynomials $E = \{P_1,\ldots,P_m\}$ such that for each $j = 1,\ldots,m$, $P_j \in \mathbb{F}_2[x_1,\ldots,x_n]$
is given by its algebraic normal form (ANF) as a sum of monomials.
The goal is to find a satisfying assignment to the system, namely
$\assign{x} = (\assign{x}_1, \ldots, \assign{x}_n) \in \{0,1\}^n$ such that $P_j(\assign{x})=0$
for every $j \in \{1,\ldots,m\}$,
or to determine that such an assignment does not exist. An additional important parameter of a polynomial system is its degree $d$ which bounds the algebraic degree of its polynomials. Typically, $d$ is assumed to be a small constant.

If all polynomials in $E$ are linear (i.e., $d=1$), then the system can be efficiently solved (e.g., by Gaussian elimination). However, the problem is known to be NP-hard already for quadratic systems (namely, $d=2$). Moreover, assuming the exponential time hypothesis~\cite{ImpagliazzoP01}, there exists no subexponential time (worst-case) algorithm for this problem.
Yet, devising the most efficient exponential time algorithm for solving polynomial systems over $\mathbb{F}_2$ (and quadratic systems over $\mathbb{F}_2$ in particular) is an interesting and active research problem. It is particularly relevant to the domain of cryptography, as the security of various cryptosystems (known as multivariate cryptosystems) is directly based on the conjectured hardness of solving these systems. Examples of such ciphers include HFE by Patarin~\cite{Patarin96}, UOV by Kipnis, Patarin and Goubin~\cite{KipnisPG99}, and more recent cryptosystems such as GeMSS~\cite{Casanova17} which is a second round candidate signature scheme in NIST's post-quantum standardization project~\cite{NIST}.

\subsection{Previous Work}

Several restricted classes of non-linear equation systems over $\mathbb{F}_2$ are known to have polynomial time algorithms. Examples include extremely under-determined ($n > m(m+1)/2$) or over-determined ($m > n(n+1)/2$) quadratic systems (in the latter case an efficient algorithm -- based on Gaussian elimination -- exists for most instances). Various works investigate algorithms that interpolate between one of these extreme cases and the case of $m = n$ (e.g.,~\cite{ThomaeW12} focuses on under-determined systems).

\paragraph{Algebraic techniques.} A common approach to the general problem of solving polynomial systems over finite fields, is to use algebraic techniques in order to find a convenient representation of the ideal generated by the polynomials in the form of a reduced Gr\"{o}bner basis. Well-known algorithms for computing Gr\"{o}bner bases include Buchberger’s algorithm~\cite{Buchberger65}, F4~\cite{faugere99} and F5~\cite{Faugere02}, but the asymptotic complexity of these algorithms is not well-understood. At a high level, Gr\"{o}bner basis algorithms and their variants represent the polynomials in the input system, along with many of their multiples, using a matrix and then attempt to reduce it via elimination techniques. Another algorithm that employs related methods is the XL algorithm~\cite{CourtoisKPS00} which was developed for cryptanalytic purposes and has led to several variants. In particular, in~\cite{YangC04} Yang and Chen developed an XL variant whose complexity for solving quadratic equations over $\mathbb{F}_2$ with $m=n$ was shown to be $O(2^{0.875n})$ under some algebraic assumptions.

In~\cite{BardetFSS13}, Bardet et al. devised an algorithm for solving quadratic polynomial systems over $\mathbb{F}_2$ based on a hybrid approach that combines exhaustive search over a subset of the variables with elimination techniques. Under some algebraic assumptions on the input system which were experimentally found to hold for random systems, the authors bounded the complexity of the deterministic variant of their algorithm by $O(2^{0.841n})$, while the expected complexity of the Las Vegas variant was bounded by $O( 2^{0.792n})$. More recently, Joux and Vitse developed a new algorithm based on a different hybrid approach and showed that it outperforms in practice previous algorithms for a wide range of parameters~\cite{JouxV17}. However, analyzing the asymptotic complexity of this algorithm seems difficult and is a very interesting open problem.

\paragraph{The polynomial method.} In~\cite{LokshtanovPTWY17} Lokshtanov et al. presented the first worst-case algorithms for solving polynomial equations over finite fields with exponential speedup over exhaustive search. In particular, their randomized algorithm for solving equations over $\mathbb{F}_2$ has runtime of $O( 2^{0.8765n})$ for quadratic systems and $O(2^{(1- 1/(5d)) n})$ in general. In more recent work, Bj\"{o}rklund, Kaski and Williams~\cite{BjorklundK019} revisited the algorithms of Lokshtanov et al. for polynomial systems over $\mathbb{F}_2$ and improved their complexity to $O(2^{0.804n})$ for $d=2$ and $O(2^{(1- 1/(2.7d)) n})$ in general.

These new algorithms are based on the so-called \emph{polynomial method} in circuit complexity~\cite{Beigel93} that has been recently applied in algorithm design~\cite{Williams14}.
At a high level, the new algorithms represent the system $E$ by the single polynomial over $\mathbb{F}_2$, $$F(x) = (1+P_1(x))(1+P_2(x)) \ldots (1+P_m(x))$$
that evaluates to 1 (only) on solutions to $E$. However, the ANF of $F(x)$ generally has a huge number of terms and hence is replaced with a probabilistic polynomial which agrees with it on most assignments, but its ANF has a smaller number of terms. The probabilistic polynomial is then evaluated efficiently on many carefully chosen assignments (using fast multipoint evaluation techniques), leading to an exponential advantage over exhaustive search.

\subsection{Our Results}

In this paper we present a randomized worst-case algorithm for solving polynomial systems of degree $d$ over $\mathbb{F}_2$ with better asymptotic complexity than all previously published algorithms.
\begin{theorem}
\label{thm:main}
There is a randomized algorithm that given a system $E$ of polynomial equations over $\mathbb{F}_2$ with degree at most $d$ in $n$ variables, finds a solution to $E$ or correctly decides that a solution does not exist with high probability. The runtime of the algorithm is bounded by $O(2^{0.6943n})$ for $d=2$ and by $O(2^{(1 - 1/(2d))n})$ for $d > 2$.
\end{theorem}
We note that for $d=2$, the complexity of our algorithm can be made arbitrarily close to $O(\varphi^n)$, where $\varphi = \tfrac{1}{2}(1 + \sqrt{5})$ is the golden ratio. Furthermore, the stated complexity bound for $d > 2$ is somewhat loose. A more precise (but rather unwieldy) complexity bound for our algorithm is given in Theorem~\ref{thm:main2}. For example, using the precise formula, we bound the complexities of our algorithm for $d=3$ and $d=4$ by $O(2^{0.8114n})$ and $O(2^{0.8633n})$, respectively.

In addition, we consider the problem of outputting all the solutions to a given system.
\begin{theorem}
\label{thm:exhaust}
There is a randomized algorithm that given a system $E$ of polynomial equations over $\mathbb{F}_2$ with degree at most $d$ in $n$ variables, outputs all $K$ solutions to $E$ or correctly decides that a solution does not exist with high probability. For an arbitrarily small $\epsilon > 0$, the runtime of the algorithm is bounded by
$$O \left(\max \left( 2^{0.6943n}, K \cdot 2^{\epsilon n} \right) \right) \text { for } d=2,
\text{ and by }
O \left( \max \left( 2^{(1 - 1/(2d))n}, K \cdot 2^{\epsilon n}\right) \right) \text { for } d>2.$$
\end{theorem}
When $K \geq 2^{0.6943n}$ for $d=2$ and $K \geq 2^{(1 - 1/(2d))n}$ for $d> 2$, then our algorithm is close to the best possible (assuming the goal is to explicitly output solutions, rather than some compact representation of them).

The algorithm may be useful in various scenarios. For example, suppose the system $E$ contains several polynomials of degree $d=2$ and additional ones with higher degrees. Then, one may attempt to find a solution to $E$ by enumerating all solutions to the system $E'$ which contains only the quadratic polynomials, and testing them on the remaining polynomials. Denoting the number of solutions to $E'$ by $K'$, the approach is generally preferable compared to analyzing the entire system at once if $K' \leq 2^{0.6943n}$ (this may depend on additional parameters otherwise). Moreover, by random sampling, one can obtain a sufficiently good estimate of $K'$ in order to determine the complexity in advance.

The improvement of our algorithms compared to the state-of-the-art is significant and surprising to a certain extent. To demonstrate this, we consider the work of Bj\"{o}rklund and Husfeldt~\cite{BjorklundH13} which presented an algorithm for counting the parity of the number of Hamiltonian cycles in a directed graph on $n$ vertices.\footnote{We thank Andreas Bj\"{o}rklund for pointing out the connection between this work and~\cite{BjorklundH13}.} Their algorithm related the problem of Hamiltonian cycle parity-counting to the problem of listing all solutions to a structured quadratic system over $\mathbb{F}_2$ with $n$ variables. The system is succinctly expressed as $x \circ Ax = x$, where $A$ is the adjacency matrix of the graph and the operator $\circ$ denotes the coordinate-wise product. Bj\"{o}rklund and Husfeldt bounded the number of solutions to this system by $1.5^{n}$ and devised an algorithm for listing all of them in time $O(\varphi^n)$, which gives the total runtime of the Hamiltonian cycle parity-counting algorithm. Interestingly, the asymptotic complexity of our algorithm for $d=2$ in Theorem~\ref{thm:exhaust} comes arbitrarily close to the complexity of the algorithm of~\cite{BjorklundH13} which was tailored to list solutions to equation systems of a very specific form. This also implies that any further (exponential) improvement to our algorithm would improve the algorithm for Hamiltonian cycle parity-counting as well.

\subsection{Techniques}

We continue the line of work on polynomial method-based algorithms for solving equation systems, initiated by Lokshtanov et al.~\cite{LokshtanovPTWY17} and Bj\"{o}rklund et al.~\cite{BjorklundK019}. In particular, we revisit the algorithm of Bj\"{o}rklund et al. which reduced the problem of finding a solution to a polynomial system over $\mathbb{F}_2$ to a \emph{parity-counting} problem of computing the overall parity of solutions $\sum_{\assign{x} \in \{0,1\}^n} F(\assign{x})$ for $F(x) = (1+P_1(x)) \ldots (1+P_m(x))$. By exploiting probabilistic polynomials that approximate $F(x)$, Bj\"{o}rklund et al. reduced a parity-counting instance to many smaller instances of the same problem, where each smaller instance was obtained by fixing a variable subset to a particular value.

Our main observation is that all of these smaller parity-counting instances are related and solving them independently is suboptimal. In order to exploit this observation, we define a new problem of \emph{multiple parity-counting} which solves many small parity-counting instances at once. Making use of probabilistic polynomials, we reduce an instance of multiple parity-counting to only a few instances of the same problem.

Interestingly, we further use our multiple parity-counting algorithm in a natural way as a sub-procedure in an algorithm that enumerates over all solutions to a polynomial system (Theorem~\ref{thm:exhaust}) with overhead that is essentially optimal. On the other hand, it is not obvious how to output an exponential number of solutions using the related algorithms of~\cite{BjorklundK019,LokshtanovPTWY17} without increasing their complexity proportionally.

We describe some preliminaries (including the Bj\"{o}rklund et al. algorithm~\cite{BjorklundK019}) next, while the proofs of theorems~\ref{thm:main} and~\ref{thm:exhaust} are given in sections~\ref{sec:improve} and~\ref{sec:exhaust}, respectively.

\section{Preliminaries}
\label{sec:prelim}

Given a finite set $S$, denote by $|S|$ its size. Given a vector $x \in \{0,1\}^n$, let $\mathrm{HW}(x)$ denote its Hamming weight. We denote by $W^{n}_w$ the set
$\{x \in \{0,1\}^{n} \mid \mathrm{HW}(x) \leq w\}$. Note that $|W^{n}_w| = \sum_{i=0}^{w} \tbinom{n}{i}$ and we simplify notation by denoting $\lbinom{n}{w} = \sum_{i=0}^{w} \tbinom{n}{i}$. We denote by $W^{n}_{w}[i]$ the $i$'th element of the set $W^{n}_{w}$ (according to some fixed ordering of its elements). We further denote by $[n]$ the set $\{1, \ldots, n \}$.

For $0\leq q \leq 1$, let $\HH(q) = -q \log q - (1-q) \log (1-q)$ be the binary entropy function.
We will use the following well-known property of this function.
\begin{fact}
\label{fact:entropy1}
For positive integers $u,v$ such that $v \leq \tfrac{u}{2}$,
\begin{align*}
\tfrac{1}{u+1} \cdot 2^{u \HH \left(\tfrac{v}{u} \right)} \leq \tbinom{u}{v} \leq 2^{u\HH \left(\tfrac{v}{u} \right)}.
\end{align*}
\end{fact}
\begin{definition}
\label{def:f}
Let
$f_d(p)= (1 - p) \cdot \HH\left(\tfrac{(d-1) p}{1-p} \right),$
and denote by $\tau(d)$ the maximal value of this function in the interval $\left[0,\tfrac{1}{2d-1} \right]$.
\end{definition}
We need the following fact.
\begin{fact} [Analysis of $\tau(d)$]
\label{fact:tau}
$\tau(d)$ satisfies the following properties:
\begin{enumerate}
  \item $\tau(2) = \log(\varphi) < 0.6943$, where $\varphi = \tfrac{1}{2}(1 + \sqrt{5})$ is the golden ratio.
  \item For any $d > 1$, $1 - \tfrac{1}{2d-1} \leq \tau(d) < 1 - \tfrac{1}{2d}$.
\end{enumerate}
\end{fact}

One can also generalize the first property and show that for all integers $d \geq 2$,
$\tau(d) = (d-1) \log r$, where $r$ is the real positive root of the polynomial $x^d - x - 1 = 0$.

\begin{proof}
The computation of $\tau(2)$ can be done by standard analysis of the function $f_2(p)$. Alternatively,
it can be deduced from Fact~\ref{fact:entropy1} applied with $u = n-i$ and $v = i$ (where $n \rightarrow \infty$),
combined with the equality
$\sum_{i=0}^{\lfloor (n-1)/2 \rfloor} \tbinom{n-i-1}{i} = F_n$, where $F_n$ is the $n$'th Fibonacci number.

The lower bound on $\tau(d)$ for arbitrary $d > 1$ is obtained by noting that $f_d(\tfrac{1}{2d-1}) = 1 - \tfrac{1}{2d-1}$.
For the upper bound, we use the inequality $\HH(q) \leq 2\sqrt{q(q-1)}$ and deduce
$f_d(p) \leq 2\sqrt{(d-1)p(1 - dp)}$.
As the maximal value of the right hand side is $\sqrt{\tfrac{d-1}{d}}$, obtained at $p = \tfrac{1}{2d}$,
we get $\tau(d) \leq \sqrt{\tfrac{d-1}{d}}$.
Using the inequality $\sqrt{d} - \sqrt{d-1} > \tfrac{1}{2\sqrt{d}}$ (for any $d > 1$),
we obtain $1 - \tau(d) \geq \tfrac{\sqrt{d} - \sqrt{d-1}}{\sqrt{d}} > \tfrac{1}{2d}$.
\end{proof}

\subsection{Boolean Algebra}
\label{subsec:boolean}

Any function $F :\mathbb{F}_2^n \rightarrow \mathbb{F}_2$ can be described as a multilinear polynomial, whose algebraic normal form (ANF) is unique and given by $F(x_1,\ldots,x_n)=\sum_{u \in \{0,1\}^n} \alpha_u(F) M_u(x)$, where $\alpha_u(F) \in \{0,1\}$ is the coefficient of the monomial $M_u(x) = \prod_{i=1}^{n} x_i^{u_i}$ (the sum and multiplication are over $\mathbb{F}_2$). We differentiate between formal (vectors of) variables and assignments to these variables using the following notation: an assignment to the formal variable vector $x$ in $F(x)$ is denoted by $\assign{x}$ and the value of $F$ on this assignment is $F(\assign{x})$.

The algebraic degree of the function $F$ is defined as $\max\{\mathrm{HW}(u) \mid \alpha_u(F) \neq 0\}$. Therefore, a function $F$ with a degree bounded by $d \leq n$ can be described using $\lbinom{n}{ d}$ coefficients.

\paragraph{Interpolation.}
Any ANF coefficient $\alpha_u(F)$ can be interpolated by summing (over $\mathbb{F}_2$) over $2^{\mathrm{HW}(u)}$ evaluations of $F$: for $u \in \{0,1\}^n$, define
$I_{u} = \{i \in [n] \mid u_i = 1 \}$ and let
$S_u = \{x \in \{0,1\}^n \mid I_{x} \subseteq I_u \}$. Then,
\begin{align}
\label{eq:inter}
\alpha_u(F) = \sum_{\assign{x} \in S_u}F(\assign{x}).
\end{align}
Indeed, among all monomials only $M_u(\assign{x})$ attains a value of 1 an odd number of times in the expression
$$\sum_{\assign{x} \in S_u}F(\assign{x}) = \sum_{\assign{x} \in S_u} \sum_{v \in \{0,1\}^n} \alpha_v(F) M_v(\assign{x}).$$
Therefore, a function $F$ of degree bounded by $d \leq n$ can be fully interpolated from its evaluations on the set $W^{n}_d$.

\begin{fact}[Symbolic interpolation]
\label{fact:inter}
Let $F :\mathbb{F}_2^n \rightarrow \mathbb{F}_2$. For some $1 \leq n_1 \leq n$, partition its $n$ variables into two sets $y_1,\ldots,y_{n-n_1},z_1,\ldots,z_{n_1}$. Given the $\mathrm{ANF}$ of $F$, factor out all the monomials that are multiplied with $z_1 \ldots z_{n_1}$, and write it as
$F(y,z) = (z_1 \ldots z_{n_1}) F_1(y) + F_2(y,z),$
where each monomial in $F_2(y,z)$ misses at least one variable from $\{ z_1,\ldots,z_{n_1} \}$.
Then,
\begin{align*}
F_1(y) = \sum_{\assign{z} \in \{0,1\}^{n_1}} F(y,\assign{z}).
\end{align*}
\end{fact}
The fact follows from~(\ref{eq:inter}) by considering the polynomial $F_1(y)$ as the symbolic coefficient of the monomial $z_1 \ldots z_{n_1}$. Note that if $F(y,z)$ is of degree $d$ then $F_1(y)$ is of degree at most $\max(d - n_1,0)$.

\paragraph{The M\"{o}bius transform.}
Given the truth table of an arbitrary function $F :\mathbb{F}_2^n \rightarrow \mathbb{F}_2$ (as a bit vector of $2^n$ entries), the ANF of $F$ can be represented as a bit vector of $2^n$ entries, corresponding to its $2^n$ coefficients.
It follows from~(\ref{eq:inter}) that the ANF representation can be computed from the truth table of $F$ via a linear transformation over $\mathbb{F}_2$. This linear transformation is known as (a specific type of) the \emph{M\"{o}bius transform} and it can be computed in $O(n \cdot 2^n)$ bit operations.
The M\"{o}bius transform over $\mathbb{F}^n_2$ coincides with its inverse which corresponds to evaluating the ANF representation of $F$ (i.e., computing its truth table).

When dealing with a function $F$ of low degree $d$, the M\"{o}bius transform allows to convert the evaluations of $F$ on the set $W^{n}_d$ to its ANF representation (and vise-verse) and it can be computed  in time $O\left(n \lbinom{n}{ d} \right)$. Even more generally, given the ANF representation of a function $F$ of degree $d$ such that its variables are partitioned into 2 sets, $(y,z) =  (y_1,\ldots,y_{n-n_1},z_1,\ldots,z_{n_1})$, we can evaluate this polynomial on the set of points $(y,z) \in W^{n-n_1}_{d_1} \times \{0,1\}^{n_1}$ (for $d_1 \geq d$) in time
$$O(n \cdot |W^{n-n_1}_{d_1} \times \{0,1\}^{n_1}|) = O \left(n \cdot 2^{n_1} \lbinom{n-n_1}{ d_1} \right).$$
This evaluation is performed by composing the M\"{o}bius transforms on $W^{n-n_1}_{d_1}$ and $\{0,1\}^{n_1}$. For more details on this transform, refer to~\cite{KaskiKW16}.

\subsection{Solving Polynomial Systems over $\mathbb{F}_2$}
\label{sec:solving}

In this paper we deal with the following problem.
\begin{definition}[Solving a polynomial system over $\mathbb{F}_2$]
The input to the problem of solving a polynomial system over $\mathbb{F}_2$ consists of a system of $m$ polynomial equations of degree $d$ in the $n$ Boolean variables $x_1,\ldots,x_{n}$, denoted by $E = \{P_j(x)\}_{j = 1}^{m}$,
where each $P_j \in \mathbb{F}_2[x_1,\ldots,x_n]$ is given by its ANF.
A vector $\assign{x} \in \{0,1\}^n$ is solution to $E$ if $P_j(\assign{x}) = 0$ for all $j \in \{1,\ldots,m\}$. The problem has three variants:
\begin{enumerate}
  \item Decision: the output is Boolean and defined to be 1 if and only if $E$ has a solution.
  \item Search: the output is any (single) solution if $E$ is solvable, and NULL otherwise.
  \item Exhaustive: the output consists of all solutions to $E$.
\end{enumerate}
\end{definition}
In this paper we only consider randomized (Monte Carlo) algorithms for these problems. We will assume that $m \leq \lbinom{n}{ d}$ (and thus is polynomial in $n$). Note that if $m > \lbinom{n}{ d}$ then $E$ must contain linearly dependent equations that can be removed by Gaussian elimination (or if the equations are inconsistent, the system is unsolvable).

As noted by Bj\"{o}rklund et al. in~\cite{BjorklundK019}, the search variant reduces to the decisional variant: assuming the system has a solution, we first solve the decisional problem with $\assign{x}_1 = 0$ and with $\assign{x}_1 = 1$ and fix $x_1$ to a value for which a solution exists (with sufficiently high probability). Iteratively fixing all the variables gives a solution after at most $2n$ calls to the decision algorithm.

\subsection{Probabilistic Polynomials}

Given $m$ polynomial equations of degree $d$ in the $n$ Boolean variables $x_1,\ldots,x_{n}$, $E = \{P_j(x)\}_{j =1}^{m}$, consider the polynomial
\begin{align}
\label{eq:mainpoly}
F(x) = (1+P_1(x))(1+P_2(x)) \ldots (1+P_m(x)).
\end{align}
Note that $\assign{x}$ is a solution to $E$ if and only if $F(\assign{x}) = 1$. However, the degree of $F(x)$ is $d \cdot m$ in general, and its ANF may be too large to manipulate.

A key idea in the algorithm of Lokshtanov et al.~\cite{LokshtanovPTWY17}, and then in the followup work of Bj\"{o}rklund et al.~\cite{BjorklundK019} is the use of probabilistic polynomials that approximate $F(x)$ and have a smaller degree.
In particular, these works use the following construction (generally credited to Razborov~\cite{Razborov87} and Smolensky~\cite{Smolensky87}).
Let $\ell < m$ be a parameter. For $i \in \{1,\ldots,\ell\}, j \in \{1,\ldots,m\}$, pick $\alpha_{ij} \in \{0,1\}$ uniformly at random and define $\ell$ degree $d$ polynomials as
\begin{align*}
R_i(x) = \sum_{j =1}^{m} \alpha_{ij} P_j(x).
\end{align*}
For any $\assign{x} \in \{0,1\}^n$, if $F(\assign{x}) = 1$, then $R_i(\assign{x}) = 0$, whereas if $F(\assign{x}) = 0$, then there exists $j$ such that $P_j(\assign{x}) = 1$ and therefore $\Pr[R_i(\assign{x}) = 1] = \frac{1}{2}$ (the probability is over $\{\alpha_{ij}\}_{j = 1}^{m}$).
Let
\begin{align}
\label{eq:probPoly}
\tilde{F}(x) = (1+R_1(x))(1+R_2(x)) \ldots (1+R_\ell(x)).
\end{align}
By the above property, we get the following fact.
\begin{fact}
\label{fact:prob}
Let $F(\assign{x})$ and $\tilde{F}(\assign{x})$ be defined as in~(\ref{eq:mainpoly}) and~(\ref{eq:probPoly}), respectively.
For any $\assign{x} \in \{0,1\}^n$, if $F(\assign{x}) = 1$ (equivalently, $\assign{x}$ is a solution to $E$) then $\tilde{F}(\assign{x}) = 1$, whereas if $F(\assign{x}) = 0$ (equivalently, $\assign{x}$ is not a solution to $E$) then
\begin{align*}
\Pr[\tilde{F}(\assign{x}) = 0] = 1 - 2^{-\ell}.
\end{align*}
\end{fact}
Note that the degree of $\tilde{F}(x)$ is at most $d \cdot \ell$, which may be much lower than the degree of $F(x)$.

Next, we describe how such polynomials are used in the algorithm of Bj\"{o}rklund et al.

\subsection{The Bj\"{o}rklund et al. Algorithm for Solving Polynomial Systems~\cite{BjorklundK019}}
\label{sec:reduction}

The Bj\"{o}rklund et al. algorithm is based on a reduction to the parity-counting problem, as defined below.

\begin{definition}[Parity-counting problem]
The input to the parity-counting problem consists of a system of $m$ polynomial equations of degree $d$ in the $n$ Boolean variables $x_1,\ldots,x_{n}$, denoted by $E = \{P_j(x)\}_{j = 1}^{m},$
where each $P_j \in \mathbb{F}_2[x_1,\ldots,x_n]$ is given by its ANF.
Let $F(x) = (1+P_1(x)) \ldots (1+P_m(x))$.
The output is the overall parity of solutions $\sum_{\assign{x} \in \{0,1\}^n} F(\assign{x})$.
\end{definition}

\subsubsection{Reduction from decisional polynomial system solving to parity-counting.}

The algorithm of~\cite{Smolensky87} uses the Valiant-Vazirani affine hashing~\cite{ValiantV86} in order to reduce the decisional problem of solving a polynomial system to several calls (whose number is polynomial in $n$) to an algorithm for the parity-counting problem.
We briefly sketch this reduction below.

Obviously, if $\sum_{\assign{x} \in \{0,1\}^n} F(\assign{x}) = 1$, then the system represented by $E$ has a solution, but the opposite direction does not hold in general. The main idea (borrowed from~\cite{ValiantV86}) is to add several random affine equations to the system with the goal of isolating some solution $\assign{x}$ (namely, $\assign{x}$ will be the only solution to the extended system), ensuring that the value of the parity-counting problem on this instance is 1. The number of equations that we need to add in order to guarantee success with high probability depends on the (base 2) logarithm of the number of solutions to $E$, denoted by $k$, which is generally unknown. Yet, we can exhaust all $n+1$ possibilities of $k = 0,1,\ldots,n$.

\subsubsection{The Bj\"{o}rklund et al. parity-counting algorithm.}

We summarize the Bj\"{o}rklund et al. parity-counting algorithm. For more details and analysis, refer to~\cite{BjorklundK019}.

Define the probabilistic polynomial $\tilde{F}$ as in~(\ref{eq:probPoly}).
To exploit its properties, partition the $n$ variables into 2 sets $y = y_1,\ldots,y_{n-n_1}$ and $z = z_1,\ldots z_{n_1}$, where $n_1 < n$ is a parameter.
Let $G(y) = \sum_{\assign{z} \in \{0,1\}^{n_1}} \tilde{F}(y,\assign{z})$. Writing $\tilde{F}(y,z) = (z_1 \ldots z_{n_1}) \cdot \tilde{F}_1(y) + \tilde{F}_2(y,z)$, by Fact~\ref{fact:inter}, $G(y) = \tilde{F}_1(y)$ and its degree is at most $d \cdot \ell - n_1$.\footnote{In the previous algorithm of Lokshtanov et al.~\cite{LokshtanovPTWY17}, a similar polynomial to $G$ was defined, but it had a higher degree which resulted in a less efficient algorithm.} Then, interpolate $G(y)$ (as described at the end of this section) and evaluate it on all $\assign{y} \in \{0,1\}^{n - n_1}$.
For each such $\assign{y}$, by Fact~\ref{fact:prob} and a union bound over all $\assign{z} \in \{0,1\}^{n_1}$,
\begin{align*}
\Pr \left[ G(\assign{y}) = \sum_{\assign{z} \in \{0,1\}^{n_1}} F(\assign{y},\assign{z}) \right] =
\Pr \left[\sum_{\assign{z} \in \{0,1\}^{n_1}} \tilde{F}(\assign{y},\assign{z}) = \sum_{\assign{z} \in \{0,1\}^{n_1}} F(\assign{y},\assign{z}) \right]  \geq
1 - 2^{n_1 - \ell}.
\end{align*}
Choose $\ell = n_1 + 2$, so the computed \emph{partial parity} $\sum_{\assign{z} \in \{0,1\}^{n_1}} \tilde{F}(\assign{y},\assign{z})$ is correct with probability at least $\frac{3}{4}$. For each $\assign{y} \in \{0,1\}^{n - n_1}$ the error is reduced similarly to~\cite{LokshtanovPTWY17}: compute $t = \Theta(n)$ independent probabilistic polynomials $\{G^{(k)}(y)\}_{k=1}^{t}$ to obtain $t$ approximations of each partial parity and maintain a scoreboard of ``votes'' for it. Then, take a majority vote for each $\assign{y} \in \{0,1\}^{n - n_1}$ across all $t$ approximations to obtain the true partial parity, except with exponentially small probability.

Assuming all true partial parities $\sum_{\assign{z} \in \{0,1\}^{n_1}} F(\assign{y},\assign{z})$ are correctly computed, output the total parity
$$\sum_{\assign{x} \in \{0,1\}^n} F(\assign{x}) = \sum_{\assign{y} \in \{0,1\}^{n - n_1}} \sum_{\assign{z} \in \{0,1\}^{n_1}} F(\assign{y},\assign{z}).$$

\paragraph{Interpolating $G(y)$.}
We have
\begin{align*}
G(\assign{y}) = \sum_{\assign{z} \in \{0,1\}^{n_1}} \tilde{F}(\assign{y},\assign{z}) = \\
\sum_{\assign{z} \in \{0,1\}^{n_1}} (1+R_1(\assign{y},\assign{z})) \ldots (1+R_\ell(\assign{y},\assign{z})) =
\sum_{\assign{z} \in \{0,1\}^{n_1}} (1+R_{1|\assign{y}}(\assign{z})) \ldots (1+R_{\ell|\assign{y}}(\assign{z})),
\end{align*}
where $R_{i|\assign{y}}(\assign{z}) = R_{i}(\assign{y},\assign{z})$. Therefore, each evaluation $G(\assign{y})$ reduces to solving a parity-counting instance for the system $\{R_{i|\assign{y}}(z)\}_{j =1}^{\ell}$, which has $\ell$ equations of degree $d$ over $n_1$ variables. Since its degree is $d \cdot \ell - n_1$, $G(y)$ can be interpolated from its evaluations on $\{ \assign{y} \in W^{n-n_1}_{d \cdot \ell - n_1} \}$. Overall, interpolating $G(y)$ requires $\lbinom{n - n_1}{d \cdot \ell - n_1}$ recursive calls for solving (smaller) parity-counting instances.

\section{Improved Algorithm for Solving Polynomial Equation Systems over $\mathbb{F}_2$}
\label{sec:improve}

In this section we prove the following stronger variant of Theorem~\ref{thm:main} which gives a tighter bound on the runtime in terms of $\tau(d)$ (recalling Definition~\ref{def:f}).
\begin{theorem}
\label{thm:main2}
There is a randomized algorithm that given a system $E$ of polynomial equations over $\mathbb{F}_2$ with degree at most $d$ in $n$ variables, finds a solution to $E$ or correctly decides that a solution does not exist with high probability.
For an arbitrarily small $\epsilon > 0$, the runtime of the algorithm is bounded by $O(2^{(\tau(d) + \epsilon) n})$.
\end{theorem}

\begin{proof}[of Theorem~\ref{thm:main}]
By Theorem~\ref{thm:main2} for $d=2$, the complexity of the algorithm is $O(2^{(\tau(2) + \epsilon) n})$.
Using Fact~\ref{fact:tau}, we bound the complexity by $O(\varphi^{(1+\epsilon')n}) = O(2^{0.6943n})$ for sufficiently small $\epsilon'$.

For $d > 2$, by Fact~\ref{fact:tau}, $\tau(d) < 1 - \tfrac{1}{2d}$. Therefore, we can bound the complexity by
$O \left(2^{(1 - 1/(2d))n} \right)$
(since the inequality $\tau(d) < 1 - \tfrac{1}{2d}$ is strict for any $d$, we eliminate the addition of $\epsilon$ in the exponent).
\end{proof}

In the following we describe the algorithm of Theorem~\ref{thm:main2}, and bound its complexity by $O^{*}(2^{(\tau(d) + \epsilon) n})$ for an arbitrarily small $\epsilon > 0$, where the $O^{*}$ notation suppresses polynomial factors in $n$. This is the same asymptotic bound claimed in Theorem~\ref{thm:main2} (as $\epsilon > 0$ is arbitrarily small). By the reductions outlined in sections~\ref{sec:solving} and~\ref{sec:reduction}, a parity-counting algorithm gives rise to an algorithm for finding a solution to polynomial systems of degree $d$ over $\mathbb{F}_2$ (with a multiplicative polynomial overhead). Hence, we proceed to describe a parity-counting algorithm with complexity $O^{*}(2^{(\tau(d) + \epsilon) n})$. This algorithm is based on solving a somewhat more involved problem of multiple parity-counting, defined below.

\subsection{The Multiple Parity-Counting Problem}

\begin{definition}[Multiple parity-counting problem]
The input to the multiple parity-counting problem consists of a system of $m$ polynomial equations of degree $d$ in the $n$ Boolean variables $x_1,\ldots,x_{n}$, along with non-negative integers $n_1 \leq n$ and $w \leq n - n_1$.
The $n$ variables are partitioned into two sets according to $n_1$ and denoted as $y_1,\ldots,y_{n - n_1},z_{1},\ldots,z_{n_1}$, while the system is denoted by $E = \{P_j(y,z)\}_{j = 1}^{m},$
where each $P_j(y,z) \in \mathbb{F}_2[y_1,\ldots,y_{n - n_1},z_{1},\ldots,z_{n_1}]$ is given by its ANF.
Let
$F(y,z) = (1+P_1(y,z))(1+P_2(y,z)) \ldots (1+P_m(y,z)).$
The output is a vector of parities $V \in \{0,1\}^{\lbinom{n-n_1}{w}}$ such that
$$V[i] = \sum_{\assign{z} \in \{0,1\}^{n_1}} F(W^{n-n_1}_{w}[i],\assign{z}).$$
\end{definition}

We will devise an algorithm for this problem and refer to it as $\mathrm{MultParityCount}(\{P_j(y,z)\}_{j =1}^{m},n_1,w)$.
In Algorithm~\ref{alg:parity} we solve the parity-counting problem using our algorithm for the multiple parity-counting problem.

\subsubsection{Details of the multiple parity-counting algorithm.}

We describe our algorithm for the multiple parity-counting problem and give its pseudo-code in Algorithm~\ref{alg:mp}.

The algorithm begins in a similar way to the previous related algorithms~\cite{BjorklundK019,LokshtanovPTWY17}
by choosing a parameter $\ell$ and defining the probabilistic polynomial
$\tilde{F}(y,z) = (1+R_1(y,z)) \ldots (1+R_\ell(y,z))$
as in~(\ref{eq:probPoly}). Yet, we work with an additional partition of the variables.

We continue in a similar manner to the Bj\"{o}rklund et al. algorithm by partitioning the $n_1$ variables $z_1,\ldots,z_{n_1}$ into 2 sets.
Let $n_2 < n_1$ be a parameter. Let $u = u_1,\ldots,u_{n_1 - n_2}$ and $v = v_1,\ldots v_{n_2}$. Define $$G(y,u) = \sum_{\assign{v} \in \{0,1\}^{n_2}} \tilde{F}(y,u,\assign{v}).$$
Fix any $\assign{y} \in \{0,1\}^{n-n_1}$ and $\assign{u} \in \{0,1\}^{n_1 - n_2}$.
By Fact~\ref{fact:prob} and a union bound over all $\assign{v} \in \{0,1\}^{n_2}$,
\begin{align}
\label{eq:approx}
\Pr \left[G(\assign{y},\assign{u}) = \sum_{\assign{v} \in \{0,1\}^{n_2}} F(\assign{y},\assign{u},\assign{v}) \right] \geq 1 - 2^{n_2 - \ell}.
\end{align}
We choose $\ell = n_2 + 2$ as before, so each \emph{partial parity} $G(\assign{y},\assign{u}) = \sum_{\assign{v} \in \{0,1\}^{n_2}} \tilde{F}(\assign{y},\assign{u},\assign{v})$ is correct with probability at least $\frac{3}{4}$.

Writing $\tilde{F}(y,u,v) = (v_1 \ldots v_{n_2}) \cdot \tilde{F}_1(y,u) + \tilde{F}_2(y,u,v)$, by Fact~\ref{fact:inter}, $G(y,u) = \tilde{F}_1(y,u)$ and its degree is upper bounded by $d \cdot \ell - n_2$. Therefore, in order to interpolate $G(y,u)$, it is sufficient to compute its values on the set $W^{n - n_2}_{d \cdot \ell - n_2}$. Thus, for each $(\assign{y},\assign{u}) \in W^{n - n_2}_{d \cdot \ell - n_2}$, we compute $G(\assign{y},\assign{u}) = \sum_{\assign{v} \in \{0,1\}^{n_2}} \tilde{F}(\assign{y},\assign{u},\assign{v})$ and use these values to interpolate $G(y,u)$.

\paragraph{Interpolating $G(y,u)$.} The main difference from the Bj\"{o}rklund et al. parity-counting algorithm is in the way that the $|W^{n - n_2}_{d \cdot \ell - n_2}|$ evaluations of $G(y,u)$ are computed. In~\cite{BjorklundK019}, $(y,u)$ was treated as a single vector of variables $y'$ and each evaluation $G(\assign{y}') = \sum_{\assign{v} \in \{0,1\}^{n_2}} \tilde{F}(\assign{y}',\assign{v})$ was computed by a separate recursive call to the parity-counting algorithm.

On the other hand, observe that the computation of all $|W^{n - n_2}_{d \cdot \ell - n_2}| = \lbinom{n - n_2}{d \cdot \ell - n_2}$ parity-counting instances (per probabilistic polynomial)
$$\sum_{\assign{v} \in \{0,1\}^{n_2}} \tilde{F}(\assign{y},\assign{u},\assign{v}) = \sum_{\assign{v} \in \{0,1\}^{n_2}} (1+R_1((\assign{y},\assign{u}),\assign{v})) \ldots (1+R_\ell((\assign{y},\assign{u}),\assign{v}))$$
for $(\assign{y},\assign{u}) \in W^{n - n_2}_{d \cdot \ell - n_2}$ reduce to a single recursive call of the multiple parity-counting algorithm
$$\mathrm{MultParityCount}(\{R_i((y,u),v)\}_{i =1}^{\ell},n_2, d \cdot \ell - n_2).$$
We use the vector of evaluations returned from this recursive call to interpolate $G(y,u)$.

\begin{remark}
$G(y,u)$ is interpolated using its evaluations on the set $W^{n - n_2}_{d \cdot \ell - n_2} = W^{n - n_2}_{n_2 (d - 1) + 2d}$ (as $\ell = n_2 + 2$) via a call to $\mathrm{MultParityCount}$ with parameters $(n_2,n_2 (d - 1) + 2d)$, which itself calls $\mathrm{MultParityCount}$ with parameters $(n_2',n'_2 (d - 1) + 2d)$ for some $n'_2 < n_2$. Thus, the number of variables over which the polynomials are defined increases with the recursion depth, but their degree decreases (Since $d-1 \geq 1$). Our choice of parameters will ensure that $\lbinom{n-n_2}{n_2 (d - 1) + 2d} > \lbinom{n-n'_2}{n'_2 (d - 1) + 2d}$, so the new instance is not harder than the original one.
\end{remark}

\paragraph{Finalizing the algorithm.}
After interpolating $G(y,u)$, we evaluate it on all $\assign{y} \in W^{n-n_1}_{w}$ and $\assign{u} \in \{0,1\}^{n_1 - n_2}$, and obtain $\lbinom{n-n_1}{ w} \cdot 2^{n_1 - n_2}$ evaluations.

Recall that our goal is to return the true parities $\sum_{\assign{z} \in \{0,1\}^{n_1}} F(\assign{y},\assign{z})$ for each $\assign{y} \in W^{n-n_1}_{w}$.
As noted above, we choose $\ell = n_2 + 2$ and~(\ref{eq:approx}) implies that for every $(\assign{y},\assign{u})$ we have
\begin{align}
\label{eq:approx1}
\Pr \left[G(\assign{y},\assign{u}) = \sum_{\assign{v} \in \{0,1\}^{n_2}} F(\assign{y},\assign{u},\assign{v}) \right] \geq \tfrac{3}{4}.
\end{align}
Namely, we obtain the correct partial parity with probability at least $\tfrac{3}{4}$.

This allows to perform error correction using scoreboards similarly to~\cite{BjorklundK019}.
Specifically, for a parameter $t$, we compute probabilistic polynomials $\{G^{(k)}(y,u)\}_{k=1}^{t}$ and obtain $t$ approximations per $(\assign{y},\assign{u})$. We then perform a majority vote across all $t$ approximations to obtain the true partial parity $\sum_{\assign{v} \in \{0,1\}^{n_2}} F(\assign{y},\assign{u},\assign{v})$ for each $(\assign{y},\assign{u}) \in W^{n-n_1}_{w} \times \{0,1\}^{n_1 - n_2}$ (except with exponentially small probability).

Assuming we obtain the true partial parities, we can compute the required output vector of parities, as for each $\assign{y} \in W^{n-n_1}_{w}$,
$$\sum_{\assign{z} \in \{0,1\}^{n_1}} F(\assign{y},\assign{z}) = \sum_{\assign{u} \in \{0,1\}^{n_1 - n_2}} \sum_{\assign{v} \in \{0,1\}^{n_2}} F(\assign{y},\assign{u},\assign{v}).$$

\begin{algorithm}[tb]
\caption{$\mathrm{ParityCount}(\{P_j(x)\}_{j =1}^{m})$}
\label{alg:parity}
\fullhline
\begin{algorithmic}
\STATE{Parameter: $\kappa_0$}

\STATE{Initialization: $n_1 \leftarrow \lfloor \kappa_0 n \rfloor$}

\end{algorithmic}
\begin{algorithmic}[1]
\STATE{$V[0 \ldots 2^{n-n_1}-1] \leftarrow \mathrm{MultParityCount}( \{P_j(y,z)\}_{j =1}^{m}, n_1, n-n_1)$}
\STATE{$Parity \leftarrow 0$}
\FORALL{$\assign{y} \in \{0,1\}^{n-n_1}$}
   \STATE{$Parity \leftarrow Parity + V[\assign{y}]$ \com{sum is over $\mathbb{F}_2$}}
\ENDFOR
\RETURN{$Parity$}
\end{algorithmic}
\fullhline
\end{algorithm}

\begin{algorithm}[tb]
\caption{$\mathrm{MultParityCount}(\{P_j(y,z)\}_{j = 1}^{m},n_1,w)$}
\label{alg:mp}
\fullhline
\begin{algorithmic}
\STATE{Parameter: $\lambda$}
\STATE{Initialization: $n_2 \leftarrow \lfloor n_1 - \lambda n \rfloor , \ell \leftarrow n_2 + 2$, $t \leftarrow 48n +1$}

\end{algorithmic}

\begin{algorithmic}[1]

\STATE{$V[1 \ldots |W^{n-n_1}_{w}|] \leftarrow \vec{0}$
\com{initialize result array}}

\IF{$n_2 \leq 0$}
    \STATE{$V[1 \ldots |W^{n-n_1}_{w}|] \leftarrow \mathrm{BruteForceMultParity}(\{P_j(y,z)\}_{j = 1}^{m},n_1,w)$}
    \RETURN{$V$}
\ENDIF

\STATE{$SB[1 \ldots |W^{n-n_1}_{w}| \cdot 2^{n_1-n_2} ] \leftarrow \vec{0}$
\com{initialize scoreboards}}

\FORALL{$k \in \{1, \ldots, t \}$}

    \STATE{Pick $[\alpha]_{ij}^{(k)}\in \mathbb{F}_2^{\ell \times m}$ uniformly at random and compute
    $\{R_i^{(k)}(y,z)\}_{i = 1}^{\ell} = \{\sum_{j =1}^{m} \alpha_{ij}^{(k)} P_j(y,z)\}_{i = 1}^{\ell}$}

    \STATE{$V_1^{(k)}[1 \ldots |W^{n-n_2}_{d \cdot \ell - n_2}|] \leftarrow$ $\mathrm{MultParityCount}(\{R_i^{(k)}((y,u),v)\}_{i = 1}^{\ell},n_2, d \cdot \ell - n_2)$}

    \STATE{Interpolate $G^{(k)}(y,u)$: apply $\mathrm{M\ddot{o}bius}$ transform to $V_1^{(k)}[1 \ldots |W^{n-n_2}_{d \cdot \ell - n_2}|]$}

    \STATE{Evaluate $G^{(k)}(y,u)$ on $W^{n-n_1}_{w} \times \{0,1\}^{n_1 - n_2}$ by $\mathrm{M\ddot{o}bius}$ transform and store result in $Evals^{(k)}[1 \ldots |W^{n-n_1}_{w}| \cdot 2^{n_1-n_2}]$}

    \STATE{Update scoreboards $SB[1 \ldots |W^{n-n_1}_{w}| \cdot 2^{n_1-n_2} ]$ with $Evals^{(k)}$}

\ENDFOR

\FORALL{$i \in \{1, \ldots ,|W^{n-n_1}_{w}| \}$}

    \FORALL{$\assign{u} \in \{0,1\}^{n_1-n_2}$}

         \STATE{$vote \leftarrow \text{Majority}(SB[W^{n-n_1}_{w}[i],\assign{u}])$}

         \STATE{$V[i] \leftarrow V[i] + vote$ \com{sum is over $\mathbb{F}_2$}}

    \ENDFOR
\ENDFOR

\RETURN{$V$}

\end{algorithmic}
\fullhline
\end{algorithm}

\begin{algorithm}[tb]
\caption{$\mathrm{BruteForceMultParity}(\{P_j(y,z)\}_{j = 1}^{m},n_1,w)$}
\label{alg:bf}
\fullhline
\begin{algorithmic}[1]

\STATE{$Evals[1 \ldots |W^{n-n_1}_{w}| \cdot 2^{n_1}] \leftarrow \vec{1}$ \com{initialize evaluation array}}

\FORALL{$j \in \{1,\ldots, m\}$}

    \STATE{Evaluate $P_j(y,z)$ on $W^{n-n_1}_{w} \times \{0,1\}^{n_1}$ by $\mathrm{M\ddot{o}bius}$ transform and store result in $PolyEvals^{(j)}[1,\ldots,|W^{n-n_1}_{w}| \cdot 2^{n_1}]$}

    \STATE{$Evals[1 \ldots |W^{n-n_1}_{w}| \cdot 2^{n_1} ] \leftarrow$ \\ $Evals[1 \ldots |W^{n-n_1}_{w}| \cdot 2^{n_1}] \wedge PolyEvals^{(j)}[1 \ldots |W^{n-n_1}_{w}| \cdot 2^{n_1}]$ \com{bitwise AND the evaluations}}

\ENDFOR

\STATE{$V[1 \ldots |W^{n-n_1}_{w}|] \leftarrow \vec{0}$ \com{initialize result array}}

\FORALL{$i \in \{1, \ldots ,|W^{n-n_1}_{w}| \}$}

    \FORALL{$\assign{z} \in \{0,1\}^{n_1}$}

        \STATE{$V[i] \leftarrow V[i] + Evals[W^{n-n_1}_{w}[i],\assign{z}]$ \com{sum is over $\mathbb{F}_2$}}

    \ENDFOR
\ENDFOR
\RETURN{$V$}
\end{algorithmic}
\fullhline
\end{algorithm}

\subsection{Analysis}
\label{sec:analysis}

In this section we analyze Algorithm~\ref{alg:parity}, completing the proof of Theorem~\ref{thm:main2}.
Specifically, we prove the following two lemmas.
\begin{lemma}[Success probability of Algorithm~\ref{alg:parity}]
\label{lem:prob1}
For $t = 48n + 1$, Algorithm~\ref{alg:parity} is correct with probability at least $1 - 2^{-n}$.
\end{lemma}
The proof is similar to that of Bj\"{o}rklund et al.~\cite{BjorklundK019} and is given in Appendix~\ref{app:success}.

\begin{lemma}[Runtime of Algorithm~\ref{alg:parity}]
\label{lem:runtime1}
For $\lambda = \epsilon$, and $\kappa_0 = 1 - \tau(d)$, Algorithm~\ref{alg:parity} runs in time $O^{*}\left( 2^{(\tau(d) + \epsilon)n} \right)$.
\end{lemma}

\subsubsection{Runtime analysis.}

We prove Lemma~\ref{lem:runtime1}.

\begin{proof}
Denote by $T(n_1,w)$ the runtime of $\mathrm{MultParityCount}(\{P_j(y,z)\}_{j =1}^{m},n_1,w)$ (we omit the parameters $n$ and $d$ that remain unchanged in the recursive calls).
Assuming that $n_2 > 0$ and the recursive version (rather than brute force) is called,
\begin{align}
\label{eq:runtime}
\begin{split}
T(n_1,w) =
O\left( t \cdot \left( T(n_2,d \cdot \ell - n_2) +  n \cdot \lbinom{n-n_1}{w} \cdot 2^{n_1-n_2} + n \cdot \lbinom{n-n_2}{ d \cdot \ell - n_2} \right) \right) = \\
O( n \cdot T(n_2,n_2(d-1) + 2d)) + O \left( n^2 \cdot \lbinom{n-n_1}{w} \cdot 2^{n_1-n_2} \right) + O \left( n^2 \cdot \lbinom{n-n_2}{ n_2(d-1) + 2d} \right),
\end{split}
\end{align}
recalling that $\ell = n_2 + 2$, $t = 48n + 1$.

The first term corresponds to the recursive calls. The second term corresponds to the evaluations of $G^{(k)}(y,u)$ on the set $W^{n-n_1}_{w} \times \{0,1\}^{n_1 - n_2}$ using the $\mathrm{M\ddot{o}bius}$ transform, as described in Section~\ref{subsec:boolean}.
The third term corresponds to the interpolation of $G^{(k)}(y,u)$ from its values on the set $W^{n-n_2}_{d \cdot \ell - n_2}$ via the $\mathrm{M\ddot{o}bius}$ transform.

Finally, the second term dominates the runtime complexity of the remaining steps as the complexity of updating the scoreboards (and the final majority votes) is $O(|W^{n-n_1}_{w}| \cdot 2^{n_1-n_2})$. Moreover, the ANF computation of $\{R_i(y,z)\}_{i = 1}^{\ell}$ requires $O \left(t \cdot \ell \cdot m  \cdot \lbinom{n}{d} \right) = O^{*}(1)$ time (assuming $m$ is polynomial in $n$).

\paragraph{Runtime analysis by recursion level.}
We will select the parameters such that the recursion runs for a constant number of $D + 1$ levels (where $D= D(d)$) and the last level applies brute force. Note that the $i$'th level of the recursion tree (starting from the top level, where $i = 0$) contains $O(n^{i})$ nodes.
We will now start indexing the recursion variables by their level of recursion~$i$.

Let $0 < \kappa_0 \leq \tfrac{1}{2d-1}$ and $0 < \lambda < 1$ be the parameters of algorithms~\ref{alg:parity} and~\ref{alg:mp}, respectively. We denote the initial value of $n_1$ in Algorithm~\ref{alg:parity} by $n_1^{(0)} = \lfloor \kappa_0 n \rfloor$. Similarly,
$n_1^{(i+1)} = n_2^{(i)} = \lfloor n_1^{(i)} - \lambda n \rfloor$.

We have $w^{(0)} = n - n_1^{(0)}$, while for $i \geq 1$, $w^{(i)} = n_2^{(i-1)}(d-1) + 2d$. Focusing on $i \geq 1$,
\begin{align*}
\lbinom{n-n_1^{(i)}}{w^{(i)}} =
\lbinom{n-n_2^{(i-1)}}{n_2^{(i-1)}(d-1) + 2d} \leq
n^{2d} \lbinom{n-n_2^{(i-1)}}{ n_2^{(i-1)}(d-1)} \leq
n^{2d+1} \tbinom{n-n_2^{(i-1)}}{n_2^{(i-1)}(d-1)},
\end{align*}
where for the final inequality we assume that $2(n_2^{(i-1)}(d-1)) \leq n-n_2^{(i-1)}$.
Since $d-1 \geq 1$ and the sequence $n_2^{(i-1)}$ is decreasing as a function of $i$,
it is sufficient to ensure this condition for $i=1$, where it holds if
$2\kappa_0(d-1) \leq 1 - \kappa_0$, which is guaranteed since we choose $\kappa_0 \leq \tfrac{1}{2d-1}$.

Using Fact~\ref{fact:entropy1} we obtain
 \begin{align*}
\lbinom{n-n_1^{(i)}}{w^{(i)}} \leq
n^{2d+1} 2^{(n-n_2^{(i-1)})  \HH \left(\tfrac{n_2^{(i-1)}(d-1)}{n-n_2^{(i-1)}} \right)} =
n^{2d+1} 2^{n(1 - \rho_{i-1}) \HH \left(\tfrac{\rho_{i-1}(d-1)}{1-\rho_{i-1}} \right)},
\end{align*}
where we set $\rho_{i} = \tfrac{n_2^{(i)}}{n}$. Recalling that
$f_d(p) = (1 - p) \cdot \HH\left(\tfrac{(d-1) p}{1-p} \right)$
we deduce
\begin{align*}
\lbinom{n-n_1^{(i)}}{w^{(i)}} \leq n^{2d+1} 2^{f_d(\rho_{i-1})n}.
\end{align*}
We also have $n_1^{(i)} - n_2^{(i)} = n_1^{(i)} - \lfloor n_1^{(i)} - \lambda n \rfloor \leq \lambda n + 1$.
Plugging these into the second term of~(\ref{eq:runtime}), we bound it by
$$
O \left( n^{2d+3} \cdot 2^{(\lambda + f_d(\rho_{i-1}) )n}  \right)
.$$
Similarly, the third term of~(\ref{eq:runtime}) is bounded by
$
O \left( n^{2d+3} \cdot 2^{f_d(\rho_{i}) n}  \right)
$.
Since there are $O(n^i)$ nodes in the $i$'th level of the recursion, for $1 \leq i < D$ the total runtime for all nodes at the $i$'th level is bounded by
\begin{align}
\label{eq:leveli}
\begin{split}
O \left( n^{i+2d+3} \cdot ( 2^{(\lambda + f_d(\rho_{i-1}) )n} + 2^{f_d(\rho_{i}) n } ) \right) =
O^{*} \left( 2^{(\lambda + f_d(\rho_{i-1}) )n} + 2^{f_d(\rho_{i}) n}  \right),
\end{split}
\end{align}
as $i < D = O(1)$. For the root node we have
$$\lbinom{n-n_1^{(0)}}{w^{(0)}} = \lbinom{n-n_1^{(0)}}{n-n_1^{(0)}} = 2^{n-n_1^{(0)}} = 2^{n - \lfloor \kappa_0 n \rfloor} \leq 2^{1 + (1 - \kappa_0) n}.$$
Plugging this into~(\ref{eq:runtime}), its runtime is bounded by
\begin{align}
\label{eq:level0}
\begin{split}
O \left(n^2 \cdot 2^{(1 - \kappa_0 + \lambda) n} +  n^{2d+3} 2^{ f_d(\rho_{0}) n} \right) =
O^{*} \left( 2^{(1 - \kappa_0 + \lambda) n} + 2^{ f_d(\rho_{0}) n} \right).
\end{split}
\end{align}

For $i=D$, we solve the problem by brute force. By similar analysis, the runtime is bounded by
$O^{*}\left( 2^{ f_d(\rho_{D-1}) n + n_1^{(D)} }  \right)$.
Using the brute force condition $n_2^{(D)} \leq 0$, we obtain $n_1^{(D)} \leq \lambda n + 1$ and bound the runtime of each node of level $D$ by
\begin{align}
\label{eq:levellast}
O^{*}\left( 2^{(\lambda + f_d(\rho_{D-1}))n} \right),
\end{align}
and as there are $O(n^D) = n^{O(1)}$ nodes at this level, the asymptotic total runtime at this level is bounded similarly.

\paragraph{Parameter selection.}
The total runtime is determined by the runtime expressions at all levels, namely~(\ref{eq:leveli}),~(\ref{eq:level0}) and~(\ref{eq:levellast}).
Fix a value of $d > 1$ and a sufficiently small $\epsilon >0$. We will select the parameters $\kappa_0,\lambda$ such that $D = O(1)$ and the runtime is
$$O^{*}\left( 2^{(\tau(d) + \epsilon)n} \right),$$
where $\tau(d)$ is the maximum of the function $f_d(p)$
in the interval $\left[0,\tfrac{1}{2d-1} \right]$.
Note that $f_d(\tfrac{1}{2d-1}) = 1 - \tfrac{1}{2d-1}$, so $\tau(d) \geq 1 - \tfrac{1}{2d-1}$.

If we take $\lambda$ to be sufficiently small, then optimizing the parameters amounts to balancing the exponent terms $1 - \kappa_0$ in~(\ref{eq:level0}) and the remaining terms of the form $f_d(\rho_{i})$.

We choose
$\lambda = \epsilon$, $\kappa_0 = 1 - \tau(d)$. Recall that we run brute force once $n_2^{(D)} \leq  0$. Since $n_2^{(i)} \leq (\kappa_0 - (i+1) \lambda)n$, then $n_2^{(D)} \leq (\kappa_0 - (D+1) \epsilon)n$ and therefore $D \leq \tfrac{\kappa_0}{\epsilon} - 1 = O(1)$ as required.

As $\kappa_0 n \geq  n_1^{(0)} \geq n_2^{(0)} \geq \ldots \geq n_2^{(D-1)} > 0$ and
$\kappa_0 = 1 - \tau(d) \leq \tfrac{1}{2d - 1}$, then
$\rho_{i} = \tfrac{n_2^{(i)}}{n} \in [0,\tfrac{1}{2d - 1}]$ for all $i \in \{0,\ldots,D-1\}$, and therefore (by the definition of $\tau(d)$), $2^{f_d(\rho_{i})n} \leq 2^{\tau(d) n}$. Since $\lambda = \epsilon$, each of the expressions~(\ref{eq:leveli}),~(\ref{eq:level0}) and~(\ref{eq:levellast}) is bounded by $O^{*}\left( 2^{(\tau(d) + \epsilon)n} \right)$ as claimed.

Finally, it is possible to reduce the complexity to $O^{*}\left( 2^{(\tau(d) + o(1))n} \right)$. For example, by choosing $\lambda = \Theta \left(\tfrac{1}{\log n} \right)$, $D$ will no longer be constant, but the total number of recursive calls is still subexponential.
\end{proof}

\section{Exhaustively Solving Polynomial Equation Systems}
\label{sec:exhaust}

In this section we prove the following stronger variant of Theorem~\ref{thm:exhaust} which gives a tighter bound on the runtime.
\begin{theorem}
\label{thm:exhaust2}
There is a randomized algorithm that given a system $E$ of polynomial equations over $\mathbb{F}_2$ with degree at most $d$ in $n$ variables, outputs all $K$ solutions to $E$ or correctly decides that a solution does not exist with high probability. For an arbitrarily small $\epsilon > 0$, the runtime of the algorithm is bounded by
$$O \left(\max \left( 2^{(\tau(d) + \epsilon)n}, K \cdot 2^{\epsilon n} \right) \right).$$
\end{theorem}

\begin{proof}[of Theorem~\ref{thm:exhaust}]
Theorem~\ref{thm:exhaust} is obtained from Theorem~\ref{thm:exhaust2} in a similar way that Theorem~\ref{thm:main} is obtained from Theorem~\ref{thm:main2}.
\end{proof}

The algorithm's pseudo-code is given in Algorithm~\ref{alg:exhaust}. We proceed with a detailed description and analysis that bounds the complexity by $O^{*} \left(\max \left( 2^{(\tau(d) + \epsilon)n}, K \cdot 2^{\epsilon n} \right) \right)$ (which is sufficient for establishing the bound of Theorem~\ref{thm:exhaust2}).

Our approach isolates solutions similarly to the Valiant-Vazirani affine hashing~\cite{ValiantV86}. However, the affine hashing generally isolates only one solution at a time, and applying it to exhaust all solutions will be inefficient unless their number is very small. Thus, we apply a variant of the affine hashing that isolates and outputs many solutions \emph{in parallel}.

Recall that Algorithm~\ref{alg:parity} partitions the variables into two sets $(y,z)$ according to a parameter $n_1$ and calls $\mathrm{MultParityCount}( \{P_j(y,z)\}_{j =1}^{m}, n_1,w)$. If the returned parity for a specific $\assign{y} \in \{0,1\}^{n-n_1}$ is 1 and the output is correct, then there exists a solution $\assign{x} = (\assign{y},\assign{z})$ to $E$ for some unknown $\assign{z} \in \{0,1\}^{n_1}$.
We will be particularly interested in the case where there exists only one such solution.
\begin{definition}[Isolated solutions]
A solution $\assign{x} = (\assign{y},\assign{z})$ to $E = \{P_j(y,z)\}_{j = 1}^{m}$ is called isolated (with respect to the variable partition $(y,z)$), if for any $\assign{z}' \neq \assign{z}$, $(\assign{y},\assign{z}')$ is not to a solution to $E$.
\end{definition}

We first describe how to output all isolated solutions with respect to $(y,z)$ using a total of $n_1+1$ calls to $\mathrm{MultParityCount}$. We will assume that all returned parities by $\mathrm{MultParityCount}$ calls are correct (by our parameter selection, this will hold except with negligible probability).

\paragraph{Outputting isolated solutions.} After running $\mathrm{MultParityCount}$ once, let us momentarily assume that all $\assign{y} \in \{0,1\}^{n-n_1}$ for which the returned parity is 1 (namely, $\sum_{\assign{z} \in \{0,1\}^{n_1}} F(\assign{y},\assign{z}) = 1$) correspond to isolated solutions.
The remaining $n_1$ bits of these solutions can be recovered one-by-one using $n_1$ additional calls to $\mathrm{MultParityCount}$, where in call $i$, we fix variable $z_i$ to 0 in $\{P_j(y,z)\}_{j =1}^{m}$ (all calls are with respect to the same partition $(y,z)$, but $z_i$ is fixed to 0 in call $i$).

For $\assign{y} \in \{0,1\}^{n-n_1}$, $i \in \{1,\ldots,n_1\}$, $b \in \{0,1\}$, let us denote
\begin{align*}
U(\assign{y},i,b) = \sum_{ \assign{z}_1, \ldots ,\assign{z}_{i-1},\assign{z}_{i+1},\ldots,\assign{z}_{n}  \in \{0,1\}^{n_1-1}}
F(\assign{y},\assign{z}_1, \ldots ,\assign{z}_{i-1},b,\assign{z}_{i+1},\ldots,\assign{z}_{n}).
\end{align*}
By running $\mathrm{MultParityCount}$ with $z_i = 0$ we derive $U(\assign{y},i,0)$ for all $\assign{y} \in \{0,1\}^{n-n_1}$. Since
\begin{align*}
U(\assign{y},i,0) + U(\assign{y},i,1) = \sum_{\assign{z} \in \{0,1\}^{n_1}} F(\assign{y},\assign{z}),
\end{align*}
then assuming $\sum_{\assign{z} \in \{0,1\}^{n_1}} F(\assign{y},\assign{z}) = 1$, exactly one of the expressions $U(\assign{y},i,0)$ and $U(\assign{y},i,1)$ has a value of 1, and the assignment of $z_i$ in this expression is the value of $z_i$ in the isolated solution whose prefix is $\assign{y}$.

Some of the 1 parities returned by the first $\mathrm{MultParityCount}$ call may not correspond to an isolated solution, but rather to an odd number of solutions which is larger than 1. In this case, the procedure for such a $\assign{y}$ may result in a ``false positive''. Thus, we need to test that each output is indeed a solution to $E$.

\paragraph{Isolating solutions.}
It remains to describe how to isolate solutions.
For this purpose, we perform a change of variables by first selecting a uniform $n \times n$ invertible matrix $B \in \mathbb{F}_2^{n \times n}$ (e.g., by rejection sampling). We replace $x_i$ in all its occurrences in $E$ by the linear expression $\sum_{j = 1}^{n} B[i][j] v_j$ over the new variables $(v_1,\ldots,v_n)$. Note that we have $x_i = (Bv)_i$ and so $x = Bv$ as vectors of variables.

Since the change of variables is linear, the result is a system $E'$ of the same algebraic degree as $E$ over the new variables. $E'$ is equivalent to $E$ is the sense that any solution $\assign{v}$ to $E'$ corresponds to a solution $\assign{x} = B \assign{v}$ to $E$ which can be computed efficiently from $\assign{v}$ by linear algebra (and vise-versa). Thus, when we find an isolated solution $\assign{v}$ to $E'$ (with respect to some variable partition), we output $B \assign{v}$ as a solution to $E$.

For a parameter $r$, we will run the above procedure for $r$ iterations, each time performing a new and independent change of variables and outputting the isolated solutions with respect to the new variable set. Below we select the parameters and complete the analysis.

\begin{algorithm}[tb]
\caption{$\mathrm{ExhaustSolutions}(\{P_j(x)\}_{j =1}^{m})$}
\label{alg:exhaust}
\fullhline
\begin{algorithmic}
\STATE{Parameters: $n_1$}

\STATE{Initialization: $r \leftarrow 2n$}

\end{algorithmic}
\begin{algorithmic}[1]

\FORALL{$k \in \{1, \ldots, r\}$}

    \STATE{Sample uniform invertible matrix $B^{(k)} \in \mathbb{F}_2^{n \times n}$}

    \STATE{$\{Q_j^{(k)}(v)\}_{j =1}^{m} \leftarrow \mathrm{ChangeVariables}(B^{(k)}, \{P_j(x)\}_{j =1}^{m})$}

    \STATE{$ZV^{(k)}[0 \ldots n_1][0 \ldots 2^{n-n_1}-1] \leftarrow \vec{0}$ \com{init mult parity array per $z_i$}}

    \STATE{$ZV^{(k)}[0][0 \ldots 2^{n-n_1}-1] \leftarrow \mathrm{MultParityCount}( \{Q_j^{(k)}(y,z)\}_{j =1}^{m},n_1,n-n_1)$}

    \FORALL{$i \in \{1,\ldots,n_1\}$}

        \STATE{$ZV^{(k)}[i][0 \ldots 2^{n-n_1}-1] \leftarrow$ \\
         $\mathrm{MultParityCount}( \{Q_j^{(k)}(y,z_1,..,z_{i-1},0,z_{i+1},..,z_{n_1})\}_{j =1}^{m},n_1-1,n-n_1)$}

    \ENDFOR

    \FORALL{$\assign{y} \in \{0,1\}^{n-n_1}$}

        \IF{$ZV^{(k)}[0][\assign{y}] = 1$}

            \STATE{$sol \leftarrow \assign{y}$}

            \FORALL{$i \in \{1,\ldots,n_1\}$}

                \STATE{$p_0 \leftarrow ZV^{(k)}[i][\assign{y}]$}

                \IF{$p_0 = 1$}

                    \STATE{$sol \leftarrow sol \| 0$ \com{concatenate bit to solution}}

                \ELSE

                    \STATE{$sol \leftarrow sol \| 1$}

                \ENDIF

            \ENDFOR

            \IF{$B^{(k)} \cdot sol$ is a solution to $\{P_j(x)\}_{j =1}^{m}$}

                \STATE{\textbf{output} $B^{(k)} \cdot sol$}

            \ENDIF

        \ENDIF
    \ENDFOR

\ENDFOR
\end{algorithmic}
\fullhline
\end{algorithm}

\subsection{Analysis}

The following lemma completes the proof of Theorem~\ref{thm:exhaust2}.
\begin{lemma}[Analysis of Algorithm~\ref{alg:exhaust}]
\label{lem:prob2}
Let $K$ be the number of solutions to $E = \{P_j(x)\}_{j =1}^{m}$.
For $r = 2n$, there exist parameters for Algorithm~\ref{alg:parity} such that:
\begin{enumerate}
  \item It runs in time $O^{*} \left(\max \left( 2^{(\tau(d) + \epsilon)n}, K \cdot 2^{\epsilon n} \right) \right)$.
  \item It is correct with probability $1 - 2^{-\Omega(n)}$.
  \item Such parameters can be computed in time $O^{*}(2^{(1- \tau(d))n})$ with probability $1 - 2^{-\Omega(n)}$.
\end{enumerate}
\end{lemma}
Note that for any $d \geq 2$, $\tau(d) > 1- \tau(d)$, so the total complexity remains $O^{*} \left(\max \left( 2^{(\tau(d) + \epsilon)n}, K \cdot 2^{\epsilon n} \right) \right)$.

\begin{proof}

\paragraph{Preliminary runtime analysis.} The change of variables requires recomputing the ANF of all polynomials. Each polynomial in $E$ has at most $n^d$ monomials, while the substitution and ANF computation for each monomial of degree $d$ requires $O(n^d)$ time. Therefore, the total runtime of this step is $O(m \cdot n^{2d}) = O^{*}(1)$. Additional linear algebra computations also have complexity $O^{*}(1)$.
Thus, the runtime of each of the $r$ iterations is dominated by the calls to $\mathrm{MultParityCount}$.

Next, we analyze the success probability as a function of the parameters $n_1,r$.
\paragraph{Success probability analysis.} Fix a solution $\assign{x}$ to $E$. Under a change of variables $B$, it is transformed into a solution $(\assign{y},\assign{z}) = \assign{v} = B^{-1} \assign{x}$ to $E'$. We will lower bound the probability that $\assign{v}$ is isolated by the change of variables, namely, for any $\assign{z}' \neq \assign{z}$, we require that $(\assign{y},\assign{z}')$ is not to a solution to $E'$.
Let $\assign{x}'$ be another solution to $E$. Then $B^{-1} \assign{x}' = (\assign{y},\assign{z}')$ for $\assign{z}' \neq \assign{z}$ if and only if $B^{-1}(\assign{x} + \assign{x}') = (\vec{0},\assign{z} + \assign{z}')$ and $\assign{z} + \assign{z}' \neq \vec{0}$, namely, $B^{-1}(\assign{x} + \assign{x}')$ is a vector in a specific $n_1$--dimensional subspace. Since $B^{-1}$ is itself a uniform invertible linear transformation, any non-zero vector is mapped to this space with probability at most $2^{n_1 - n}$. Taking a union bound over all $K$ solutions to $E$,
\begin{align*}
\Pr[B^{-1} \assign{x} \text{ is an isolated solution to } E' \text{ with respect to } (y,z)] \geq 1 - K \cdot 2^{n_1 - n}.
\end{align*}
\paragraph{Parameter selection.} We choose $n_1$ so the isolation probability above is at least $\tfrac{1}{2}$ (except with exponentially small probability). Then, setting $r = 2n$, each solution is isolated at least once with probability at least $1 - 2^{-2n}$. Consequently, by a union bound over all $K$ solutions to $E$, all of them will be output with probability at least $1 - 2^{-n}$.

The choice of $n_1$ for which the isolation probability is $\tfrac{1}{2}$ depends on $K$ which is unknown. However, by random sampling (using a standard Chernoff bound), the fraction of solutions $\tfrac{K}{2^n}$ can be estimated up to a multiplicative factor of 2 in complexity $O^{*} \left(\tfrac{2^n}{K} \right)$ and exponentially small error probably. For $K = \Omega( 2^{\tau(d) n})$ this requires $O^{*}(2^{(1- \tau(d)) n})$ time. In particular, we calculate in time $O^{*}(2^{(1- \tau(d)) n})$ a value $\tilde{K}$ such that if $\tilde{K} \leq 2^{\tau(d)n - 2}$ then $K \leq 2^{\tau(d) n - 1}$, while if $\tilde{K} > 2^{\tau(d)n - 2}$, then $\tfrac{K}{2} \leq \tilde{K} \leq 2K$ (except with exponentially small probably).

Therefore, if $\tilde{K} \leq 2^{\tau(d)n - 2}$, we set $n_1 = \lfloor (1 - \tau(d))n \rfloor$ (as chosen in Section~\ref{sec:analysis} for Algorithm~\ref{alg:parity}) and the complexity remains $O^{*}(2^{(\tau(d) + \epsilon)n})$. Otherwise, $\tilde{K} > 2^{\tau(d)n - 2}$ and we set $n_1 = \lfloor n -  \log \tilde{K}  - 2 \rfloor$. The complexity becomes $O^{*}(2^{n - n_1 + \epsilon n}) = O^{*}(\tilde{K} \cdot 2^{\epsilon n}) = O^{*}(K \cdot 2^{\epsilon n})$ (due to the factor $O^{*}( 2^{(1 - \kappa_0 + \lambda) n} )$ in the runtime expression for the root node~(\ref{eq:level0})).
\end{proof}

\subsubsection{Acknowledgements.} The author would like to thank Andreas Bj\"{o}rklund for pointing out the connection between our algorithms for solving quadratic polynomial systems and the algorithm of Bj\"{o}rklund and Husfeldt~\cite{BjorklundH13} for counting the parity of the number of Hamiltonian cycles in a directed graph.

\bibliographystyle{splncs03}
\bibliography{equationsbib}

\appendix

\section{Success Probability Analysis of Algorithm~\ref{alg:parity}}
\label{app:success}

\begin{proof}[of Lemma~\ref{lem:prob1}]
The algorithm is guaranteed to be correct if the scoreboard majority votes are equal to the corresponding parities in the top level multiple parity-counting instance and in all the recursive calls. We choose $t$ such that each scoreboard majority is correct, except with exponentially small probability.

For $(\assign{y},\assign{u}) \in W^{n-n_1}_{w} \times \{0,1\}^{n_1-n_2}$, denote $F'(\assign{y},\assign{u}) = \sum_{\assign{v} \in \{0,1\}^{n_2}} F(\assign{y},\assign{u},\assign{v})$. Also, denote $SB[\assign{y},\assign{u}]$ the scoreboard entry for $(\assign{y},\assign{u})$ and by $M[\assign{y},\assign{u}]$ the majority for this entry.
Recall from~(\ref{eq:approx1}) that $$\Pr[G^{(k)}(\assign{y},\assign{u}) = F'(\assign{y},\assign{u})] \geq \tfrac{3}{4}$$
holds independently for each $k \in \{1,\ldots,t\}$. Since $SB[\assign{y},\assign{u}]$ is the sum of the $t$ random variables $G^{(k)}(\assign{y},\assign{u})$,
\begin{align*}
&\E[SB[\assign{y},\assign{u}] \mid F'(\assign{y},\assign{u}) = 1] \geq \tfrac{3}{4}t, \text{ and}\\
&\E[t - SB[\assign{y},\assign{u}] \mid F'(\assign{y},\assign{u}) = 0] \geq \tfrac{3}{4}t.
\end{align*}
A standard Chernoff bound for a random variable $X$ that is a sum of independent and identically distributed random variables states that for every $0 < \delta < 1$,
$$\Pr[X \leq (1 - \delta)\E[X]] \leq \exp \left(- \tfrac{\delta^2 \E[X]}{2} \right).$$
Since all random variables $\{G^{(k)}(\assign{y},\assign{u})\}_{k =1}^{t}$ are independent and identically distributed, we apply this bound with $\delta = 1/3$ and obtain
\begin{align*}
&\Pr \left[SB[\assign{y},\assign{u}] > \tfrac{t}{2} \mid F'(\assign{y},\assign{u}) = 1 \right] \geq 1 - \exp\left(- \tfrac{t}{24}\right),  \text{ and}\\
&\Pr \left[t - SB[\assign{y},\assign{u}] > \tfrac{t}{2} \mid F'(\assign{y},\assign{u}) = 0 \right] \geq 1 - \exp\left(- \tfrac{t}{24}\right),
\end{align*}
given that $t$ is odd.

For $t = 48n + 1$, we obtain
$$\Pr[M[\assign{y},\assign{u}] = F'(\assign{y},\assign{u})] \geq 1 - 2^{-2n}.$$
Taking a union bound over all scoreboard entries computed by the algorithm (whose number is smaller than $2^n$), we bound its error probability by $2^{-n}$.
\end{proof}

\end{document}